\documentclass[11pt]{article}

\usepackage{amsthm, amssymb, amsmath, amsfonts, url, enumerate, mathtools, graphics, xcolor}
\usepackage[margin=0.75in]{geometry}
\usepackage{verbatim}
\usepackage{fullpage}
\usepackage{appendix}

\newcommand{\dst}{\displaystyle}

\newtheorem{theorem}{Theorem}[section]

\newtheorem{proposition}[theorem]{Proposition}

\newtheorem{corollary}[theorem]{Corollary}
\newtheorem{lemma}[theorem]{Lemma}
\newtheorem{conjecture}[theorem]{Conjecture}

\newtheorem{example}[theorem]{Example}

\theoremstyle{definition}
\newtheorem{definition}[theorem]{Definition}

\newenvironment{prf}{\noindent{\bf Proof.~}}{\(\qed\)}
\newcommand{\BPF}{\begin{prf}} 
\newcommand {\EPF}{\end{prf}}

\def\spn{\textsf{span}}

\def\char{\textnormal{char}}

\newcommand{\Sym}{\textnormal{Sym}}

\def\poly{\textnormal{poly}}

\newcommand{\rk}{\textnormal{rank}}

\def\hrk{\textnormal{hom-rank}}
\def\smrk{\textnormal{sm-rank}}
\def\wrk{\textnormal{w-rank}}

\newcommand{\mat}{\textnormal{Mat}}
\newcommand{\ten}{\textnormal{Ten}}
\def\tennd{\ten_{n,d}(\mathbb{F})}

\newcommand{\coeff}{\textnormal{coeff}}

\def\F{{\mathbb{F}}}

\def\N{{\mathbb{N}}}
\def\R{{\mathbb{R}}}

\def\C{{\mathbb{C}}}

\def\cC{{\mathcal C}}
\def\cL{{\mathcal L}}

\def\cM{\mathcal M}

\def\cP{\textsf{P}}
\def\cNP{\textsf{NP}}
\def\cVP{\textsf{VP}}
\def\cVNP{\textsf{VNP}}

\def\ba{{\mathbf a}}
\def\bb{{\mathbf b}}

\def\be{{\mathbf e}}
\def\bf{{\mathbf f}}
\def\bg{{\mathbf g}}

\def\bo{{\mathbf 0}}
\def\bp{{\mathbf p}}
\def\bq{{\mathbf q}}

\def\bu{{\mathbf u}}
\def\bv{{\mathbf v}}

\def\bx{{\mathbf x}}
\def\by{{\mathbf y}}
\def\bz{{\mathbf z}}

\newcommand{\wt}[1]{\widetilde{#1}}

\def\D{{\partial}}

\def\then{\Rightarrow}

\def\to{\rightarrow}

\newcommand{\eps}{\epsilon}

\title{Barriers for Rank Methods in Arithmetic Complexity}

\author{
Klim Efremenko \thanks{Ben Gurion University, email: klimefrem@gmail.com.}
\and
Ankit Garg \thanks{Microsoft Research New England, email: garga@microsoft.com.}
\and
Rafael Oliveira \thanks{Department of Computer Science, University of Toronto, 
email: rafael@cs.toronto.edu.}
\and
Avi Wigderson \thanks{Institute for Advanced Study, Princeton, email: avi@math.ias.edu.}
}

\begin{document}

\maketitle

\vspace{-15pt}

\begin{abstract}

{\em Arithmetic complexity}, the study of the cost of computing polynomials via additions and multiplications, is considered (for many good reasons) simpler to understand than {\em Boolean complexity}, namely computing Boolean functions via logical gates. And indeed, we seem to have significantly more lower bound techniques and results in arithmetic complexity than in Boolean complexity. Despite many successes and rapid progress, however, foundational challenges, like proving super-polynomial lower bounds on circuit or formula size for explicit polynomials, or super-linear lower bounds on explicit 3-dimensional tensors, remain elusive.

At the same time (and possibly for similar reasons), we have plenty more excuses, in the form of ``barrier results'' for failing to prove basic lower bounds in Boolean complexity than in arithmetic complexity. Efforts to find barriers to arithmetic lower bound techniques seem harder, and despite some attempts we have no excuses of similar quality for these failures in arithmetic complexity. This paper aims to add to this study.

In this paper we address {\em rank methods}, which were long recognized as encompassing and abstracting almost all known arithmetic lower bounds to-date, including the most recent impressive successes. Rank methods (under the name of {\em flattenings}) are also in wide use in algebraic geometry for proving tensor rank and symmetric tensor rank lower bounds. Our main results are barriers to these methods. In particular, 
\begin{itemize}
\item Rank methods {\em cannot} prove better than $\Omega_d  (n^{\lfloor d/2 \rfloor})$ lower bound on the tensor rank of {\em any} $d$-dimensional tensor of side $n$. (In particular, they cannot prove super-linear, indeed even $>8n$ tensor rank lower bounds for {\em any} 3-dimensional tensors.)
\item Rank methods {\em cannot} prove $\Omega_d  (n^{\lfloor d/2 \rfloor})$ on the {\em Waring rank}\footnote{A very restricted form of depth-3 circuits}  of any $n$-variate polynomial of degree $d$. (In particular, they cannot prove such lower bounds on stronger models, including depth-3 circuits.)
\end{itemize}

The proofs of these bounds use simple linear-algebraic arguments, leveraging connections between the {\em symbolic} rank of matrix polynomials and the usual rank of their evaluations. These techniques can perhaps be extended  to barriers for other arithmetic models on which progress has halted. 

To see how these barrier results directly inform the state-of-art in arithmetic complexity we note the following.
First, the bounds above nearly match the best explicit bounds we know for these models, hence offer an explanations why the rank methods got stuck there. Second, the bounds above are  a far cry (quadratically away) from the true complexity (e.g. of random polynomials) in these models, which {\em if} achieved (by any methods), are known to imply super-polynomial formula lower bounds.

We also explain the relation of our barrier results to other attempts, and in particular how they significantly differ from the recent attempts to find analogues of ``natural proofs'' for arithmetic complexity. Finally, we discuss the few arithmetic lower bound approaches which fall outside rank methods, and some natural directions our barriers suggest.

\end{abstract}

\section{Introduction}\label{sec:intro}

Arithmetic complexity theory (often also called algebraic complexity theory) addresses the computation of algebraic objects (like polynomials, matrices, tensors) using the arithmetic field operations (and sometimes other operations like taking roots). Within computational complexity this field is nearly as old as Boolean complexity theory, which addresses the computation of discrete functions via logical operations, but of course mathematicians were interested in arithmetic computation for centuries before computer science was born. Indeed, Euclid's algorithm for computing GCD, Gauss' discovery of the FFT, and Abel's impossibility result for solving quintic equations by radicals are all precursors of arithmetic complexity theory. Today algebraic algorithms pervade mathematics! Extensive surveys of this field are presented in the books~\cite{BCS13,Gat13}, and, more focused on the present material are the  recent monographs~\cite{SY10,CKW11}, as well as the book~\cite{Lan17} which offers an algebro-geometric perspective.

Structurally, the Boolean and arithmetic theories, and especially the quest for lower bounds which we will focus on, progressed almost hand in hand. Shortly after the important discoveries of reductions and completeness leading to the definitions of $\cP$, $\cNP$, and complete problems for them, Valiant~\cite{V79} developed the arithmetic analog notions of $\cVP$, $\cVNP$ and complete problems for them. Separating these pairs of classes stand as the long-term challenges of these fields, and their difficulty has led to the study of a large variety of restricted models in both. Definitions, techniques and results have propagated back and forth and inspired progress, but, all in all, we understand the arithmetic models much better. This of course comes as no surprise. In the arithmetic setting (especially over fields that are large, of characteristic zero, or are algebraically closed) the diverse tools of algebra are available, but have no analogs in the Boolean setting. Moreover, as arithmetic computation is mostly {\em symbolic} it is (essentially) more stringent than the Boolean computation of functions\footnote{For example, the {\em polynomial} $x^p-x$ over $\F_p$ is nontrivial to compute, while the (identically zero) function it represents is trivial.}; indeed, it is known that proving (a non-uniform version of)  $\cP \neq \cNP$ implies 
$\cVP \neq \cVNP$ when the underlying field is $\C$~\cite{B13}. 
and thus arithmetic lower bounds are also formally easier to prove!

Despite exciting and impressive progress on arithmetic lower bounds (we will detail many later), some of the most basic questions remain open, and this seeming weakness of current techniques begs explanation, which will hopefully lead to new ones. In Boolean complexity there is a rich interplay between the discovery of the power of new techniques, and then their limitations, in the form of {\em barrier results}. Such results formally encapsulate a set of lower bound methods, and then prove (unconditional, or sometimes conditionally on natural assumptions) that these cannot solve basic questions. Well known barriers to large classes of techniques include the {\em relativization} barrier of Baker, Gill and Solovay~\cite{BGS75}, the {\em natural proof} barrier of Razborov and Rudich~\cite{RR94} and the {\em algebrization} barrier of Aaronson and Wigderson~\cite{AW09}. But there are many other important barriers, to more concrete lower bound methods, including~\cite{Rsub,Rapprox,Pot16}. 
Finding analogous barriers for arithmetic complexity has been much harder; while encapsulation of general lower bound techniques exists, e.g. in~\cite{G15,FSV17, GKSS17}, there are really no proofs of their limitations (we will discuss these in the related works subsection below). 

This paper provides, to the best of our knowledge, the first unconditional barrier results on a very general class of methods, capturing many of the known lower bounds, including the very exciting recent ones.
We now begin to describe, through examples, the techniques we encompass under {\em rank methods} and then explain their limitations.

\subsection{Sub-Additive Measures, Rank Bounds and Barriers}

Throughout, we will discuss the computation of multivariate polynomials over any field, by arithmetic circuits of various forms, in a way that will not necessitate too many specific details; we will give these as needed, and give formal details in the technical sections. The examples we start with below will demonstrate  many ``cheap'' computations may be encompassed by writing the output polynomial as a ``short'' sum of {\em simpler} ones. Thus lower bounds on the number of summands can yield (important) complexity lower bounds. We continue with discussing classes of such lower bound techniques, and then barrier results that put a limit on how large lower bounds such classes of techniques can prove.

\paragraph{Sub-additive measures}
Let us start with some examples and then generalize them.
\begin{itemize}
\item
One of the earliest basic results in arithmetic complexity, due to Hyafil~\cite{Hya79} states the following: if a homogeneous circuit of size $s$ computes an $n$-variate polynomial $f$ of degree $d$, then
$$ f = g_1 + g_2 + \dots + g_s$$
where each $g_i$ is {\em simple}, which here means {\em highly reducible}: 
$g_i = p_i \cdot q_i$, where the degrees of $p_i,q_i$ do not exceed $2d/3$. This result was developed towards parallelizing arithmetic computation, but can also be used for lower bounds: if we could find any sub-additive measure $\mu$ on polynomials, which is small on all possible $g_i$ but is large on $f$, we would have a lower bound on the minimum circuit size $s$ of $f$! In particular, Hyafil's theorem implies that if the ratio of ``large'' and ``small'' values of $\mu$ is super-polynomial in $n,d$, this would imply\footnote{Since homogenous computation can efficiently simulate non-homogeneous one.}  $\cVP \neq \cVNP$! We note that Hyafil's theorem is today only one example of numerous other decomposition theorems of similar nature used in lower bounds, e.g.~\cite{N91,NW96,RY09,HWY11} to mention a few.

\item
An even simpler example, where a similar decomposition follows directly from the definition, is tensor rank. 
Assume that a $d$-dimensional tensor (with $n$ variables in each dimension) has rank $s$. This means\footnote{Directly generalizing matrix rank, which is the case  $d=2$.} that $$ f = g_1 + g_2 + \dots + g_s$$
where each $g_i$ is {\em simple}, which here means {\em of rank 1}: 
$g_i = \ell_i^{(1)} \otimes \ell_i^{(2)} \dots \otimes \ell_i^{(d)}$, where $\ell_i^{(j)}$ is a linear form in the variables of dimension $j$. Again, any sub-additive measure $\mu$ on tensors which is small on all possible rank 1 tensors $g_i$, but is large on $f$ would yield a lower bound on its tensor rank. This question is no less important than the previous one even though tensor rank seems like a more restricted complexity measure: Raz~\cite{Raz10} proved that presenting an explicit tensor $f$ of super-constant dimension $d \leq \log n / \log\log n$, with a nearly-tight tensor rank lower bound of $n^{d(1-o(1))}$ (which holds for most tensors) will imply $\cVP_e \neq \cVNP$ (namely, explicit super-polynomial lower bounds on formulas)! We note that a similar example as tensor rank, where a decomposition suggests itself by definition, is Waring rank, where each $g_i$ is a $d$-power of a linear form.
\item
A third set of examples which directly gives such decompositions of computations is when considering bounded-depth circuits. In almost all computations one can assume without loss of generality that the top (output) gate is a plus gate, and so if a polynomial $f$ is computed by a depth-$h$ circuit of size $s$, then
$$ f = g_1 + g_2 + \dots + g_s$$
where each $g_i$ is {\em simple} in being of depth $h-1$ (and moreover, with a top product gate). Sub-additive measures small on such simple polynomials and large on $f$ were the key to the many successes on remarkably tight lower bounds for depth-3 and then 
depth-4 circuits~\cite{NW96, K12, GKKS14, KLSS14, FLMS15, KS14, KS15}. These include the breakthrough of $(nd)^{\sqrt{d}}$ explicit lower bounds~\cite{GKKS14} on the size of homogeneous depth-4 circuits, which again seem much more restricted than it is: any super-constant improvement of the exponent will imply $\cVP \neq \cVNP$!
\end{itemize}

There are many other examples in which obtaining such decompositions as above uses extra tools like approximations, random restrictions, or iterations.
Abstracting all these examples and indeed most known lower bounds in arithmetic complexity\footnote{The discussion below is quite general and indeed applies to lower bounds and barriers that use sub-additive measures in practically any computational model.}, can be done in a simple way. Let $S$ be a set of {\em simple} polynomials, and let $\hat S$ be their linear span. The $S$-complexity $c_S(f)$ of a polynomial $f\in \hat S$ is simply the smallest number $s$ such that $ f = g_1 + g_2 + \dots + g_s$ and each $g_i \in S$. 
A {\em sub-additive} measure $\mu$ is a function $\mu: \hat S \rightarrow \R^+$ such that $$\mu(g+h) \leq \mu(g)+\mu(h)$$ for any $g,h \in \hat S$. Extending $\mu$ to sets, denoting $\mu(T) = \max \{\mu(g) \,: \, g\in T \}$, we can immediately derive a lower bound on $c_S(f)$ for any polynomial $f$ by
$$c_S(f) \geq \mu(f)/\mu(S).$$ 

Let $\Delta_S$ denote all possible sub-additive measures on $\hat S$.
It is a triviality that $c_S$ itself is a sub-additive measure in $\Delta_S$, and hence this method can in principle provide tight lower bound on the complexity $c_S(f)$ for every $f$. However, the difficulty of proving lower bounds precisely means that $c_S$ is hard to understand, and so we try to ``approximate it'' with simpler measures $\mu \in \Delta$ for some family $\Delta \subseteq \Delta_S$ of sub-additive measures which are hopefully simpler to understand, compute and reason about. 

\paragraph{Barriers for sub-additive measures}

This brings us to the topic of this paper: barriers, or limits to the power of such class of lower bound methods. A {\em barrier} result for any such class of sub-additive measures $\Delta \subseteq \Delta_S$ simply asserts that $\mu(f)$ is {\em small} for {\em every} 
$\mu \in \Delta$ and any $f \in \hat S$ (whenever $\mu(S)$ is small). The quantity
$$ c(\Delta) = \mu(\hat S)/\mu(S) $$
upper bounds the best lower bound which can be proven using {\em any} $\mu \in \Delta$ on {\em any} polynomial $f \in \hat S$, simply as $\mu(f) \leq c(\Delta) \cdot \mu(S)$ 
for all of them.

Of course, concrete lower bounds are obtained using specific measures $\mu$, and there is always hope that a clever variant of such a choice will give even better bounds; indeed, much of the progress in lower bounds is of this nature. The quality of barrier result is in classifying as large as possible a class of measures $\Delta$, which captures many complexity measures, such that either $c(\Delta)$ is close to the best known lower bounds, or it is well separated with a ``desired'' lower bound (e.g. one that would approach the complexity of a random polynomial, or that would significantly improve the state of art). In this paper we focus on {\em rank methods}, which we turn to describe now.

\paragraph{Rank methods}

The rank function of matrices, is at once extremely well studied and understood in linear algebra, and is sub-additive. This has made numerous (implicit and explicit) choices of sub-additive measures, for a variety of computational models, to be defined via matrix rank, as follows. Fix a field $\F$,  and let $\mat_m (\F)$  denote the set of all $m \times m$ matrices over $\F$. Fix the set of (simple) polynomials $S$, (and thereby also their span $\hat S$) as before. Define the class $\Delta_0^S \subseteq \Delta_S$ to be the set of sub-additive measures $\mu$ which arise in the following way. Let $L: S \rightarrow \mat_m(\F)$ be any {\em linear} map for some integer $m$. Namely, for all $g,h \in S$ (and hence also in $\hat S$) we have $L(g+h)=L(g)+L(h)$, and that $L(bg)=L(g)$ for any non-zero constant $b\in \F$. Define $$\mu_L(f) = rank_\F(L(f)).$$ Clearly, all these $\mu_L \in \Delta_0^S$ are sub-additive measures on $S$. We call the elements of $\Delta_0^S$ as {\em rank methods} for $S$.

As mentioned, rank methods abound in arithmetic (and other) lower bounds. The possibly familiar names including {\em partial derivatives, shifted partial derivatives, evaluation dimension, coefficient dimension} which are used e.g. in these lower bounds for monotone, non-commutative, homogeneous, multilinear, bounded-depth and other 
models~\cite{N91, S93, Rsub, NW96, K12, GKKS14, KLSS14, FSS14, FLMS15, KS14, KS15} 
are all rank methods, and in many of these papers are explicitly stated as such. Moreover, in algebraic geometry, rank methods (usually called {\em flattenings}) are responsible for almost all tensor rank and symmetric tensor rank lower bounds (see e.g.~\cite{Lan17}).

What should be stressed is that rank methods are extremely general. We do not restrict the size $m$ of matrices used in any way (and indeed in some applications, like shifted partial derivatives~\cite{GKKS14}, $m$ grows super exponentially in the basic size parameters $n,d$). Moreover, we demand no explicitness in the specification of the linear map $L$ (and indeed, in some applications, like the multilinear formula lower bounds in~\cite{R09,RY09} the map is chosen at random). The barrier results hold for all.

We prove barrier results for two classes of very weak computational models, {\em tensor rank} and {\em Waring rank}, which are very special cases of (respectively)  multilinear and homogeneous depth-3  circuits (which themselves are the weakest class of circuits studied\footnote{As depth-2 circuits simply represent polynomials trivially, as sums of monomials.}. As with all barrier results, the weaker the model for which they are proved, the better, as they scale up for stronger models automatically! As discussed above, we will compare our barriers both to the state-of-art lower bounds in these models, as well to the best one can hope for, namely the complexity of random polynomials. 

\subsection{Main results}
Our results below hold for all large enough fields $\F$ (polynomial in $n,d$). We start with tensor rank, and proceed with Waring rank, which may be viewed as a symmetric version of tensor rank. In both cases, our barrier results nearly match (up to a function of $d$, the degree\footnote{Which is a constant in the very interesting cases where the degree $d$ is a constant!}) the best explicit lower bounds (obtained by rank methods), and are roughly quadratically away from the (desired) lower bounds that hold for random polynomials.

\paragraph{Tensor rank}

Tensors abound in mathematics and physics, and have been studied for centuries. We refer the reader to the book~\cite{L12} for one good survey. From a computational perspective tensors have been extremely interesting as well, as many problems naturally present themselves in tensor form. In arithmetic complexity they are often called {\em set-multilinear} polynomials. While $2$-dimensional tensors, namely matrices, are very well understood, $d$-dimensional tensors possess far less structure, and one way this is manifested is that the problem of computing tensor rank of $3$-dimensional tensors is already $\cNP$-complete~\cite{H90}. Many special cases, approximations and related decompositions of tensors were studied, especially recently with machine learning 
applications~\cite{C96, MR06, AFH12, HK13, GM17}. 
Let us define the model and problem formally. 

Fix $n,d$. The family of polynomials of interest here is $\hat S = \ten_{n,d}(\F)$, namely degree $d$ polynomials in $d$ sets of $n$ variables (so, total of $nd$ variables), in which each monomial has precisely one variable from each set. The coefficients of  a tensor are naturally described by an $[n]^d$ box with entries from $\F$.
The simple polynomials $S$ are {\em rank-1} tensors, namely those which are products of $d$ linear forms, one in each set of variables (equivalently, the coefficients are described by the tensor product of $d$ vectors). The tensor rank of a tensor $f$ is the smallest number of rank-1 tensors which add up to it.

Most tensors have rank about $n^{d-1}/d$. Explicit lower bounds are way worse.
It is trivial to construct an explicit $d$-dimensional tensor  of rank 
$n^{\lfloor d/2 \rfloor}$, and the best known lower bound is only a factor of 
2 larger. Specifically,  
\cite{AFT11} give an explicit tensor with 0,1 coefficients of tensor rank at least 
$2n^{\lfloor d/2 \rfloor} + n - d\log n$.  Note in particular that the best lower bound for 
$d=3$ is about $3n$. Although the lower bounds of~\cite{AFT11} are not attained via a rank method, many other lower bounds for tensor rank are attained via a rank method in 
$\Delta_0^T$ ($T$ for Tensor), namely using a sub-additive measure in the class of 
rank methods~\cite{LO15, L15}. Our barrier result proves that no bound better 
$2^d \cdot n^{\lfloor d/2 \rfloor}$ can be proven by rank methods, and in particular for $d=3$, they cannot beat $8n$ (a factor $8/3$ away from the best explicit lower bound!). 

\begin{theorem}[Statement of Theorem~\ref{thm:tensor-rank}] 
$c(\Delta_0^T) \leq 2^d \cdot n^{\lfloor d/2 \rfloor}$. 
\end{theorem}

\paragraph{Waring rank}

The Waring problem has a long history in mathematics, first in its number theoretic form initiated by Waring~\cite{War70} in 1770 (writing integers as short sums of $d$-powers of other integers), and then in its algebraic form we care about, initiated by Sylvester~\cite{Syl51} in 1851 (writing polynomials as short sums of powers of linear forms). Some of the basic questions (computing this minimum for monomials and for random polynomials) were only very recently resolved, using algebraic geometric techniques~\cite{CCG12,AH95}. In arithmetic complexity this model is often referred to as {\em depth-3 powering circuits}. Let us formalize the problem.

Fix $n,d$.  The family of polynomials of interest here is $\hat S = \poly_{n,d}$, all $n$-variate polynomials of total degree $d$. The simple generating set $S$ we care about here is the set of all $d$-powers, namely
all polynomials of the form $\ell^d$, where $\ell$ is an affine function in the $n$ given variables. So, $c_S(f)$ is the smallest number $s$ such that $f$ can be written as a sum of such $d$ powers.  

For most polynomials, the Waring rank was settled by~\cite{AH95}, and is about $(n-1)^{d}$ 
for $d$ much smaller than $n$, and is precisely 
$$ \left\lceil \frac{1}{n} \cdot \binom{n+d-1}{n-1} \right\rceil. $$  
It is trivial to find an explicit $f\in \poly_{n,d}$ whose Waring rank is 
$\Omega(n^{\lfloor d/2 \rfloor})$, and  the best known lower bound, due to~\cite{GL17} 
(again via rank method in $\Delta_0^W$), is only a little better, 
$$ \binom{n+ \lfloor d/2 \rfloor - 1}{\lfloor d/2 \rfloor} + \lfloor n/2 \rfloor -1. $$
Our barrier result proves that rank methods cannot improve this lower bound even by a 
factor of roughly  $d$.

\begin{theorem}[Statement of Theorem~\ref{thm:waring-rank}] 
$c(\Delta_0^W) \leq (d+1) \cdot \binom{n+ \lfloor d/2 \rfloor}{n}. $
\end{theorem}

\subsection{High-level ideas of the proof}\label{sec:proof-overview}

As mentioned, the proofs of our barrier results use only simple tools of linear algebra (although their use and combination is a bit subtle). Here are the key ideas of the proof, written abstractly in the general notation established above (again, we believe that they can be applied in other settings beyond the two we consider in this paper).

Consider any simple set $S$ of polynomials, and rank methods $\Delta_0^S$ for it. Thus, we need to provide an upper bound on the quantity $c(\Delta_0^S)$, namely on the ratio $\mu_L(f)/\mu_L(S)$ for every $f\in \hat S$, and every linear map $L: S \rightarrow \mat_m(\F)$. Set $r = \mu_L(S)$. 

\begin{itemize}
\item
We view linear map $L$, which gives rise to a sub-additive measure in $\Delta_0^S$, as a matrix polynomial, namely as a polynomial with matrix coefficients, or equivalently as a symbolic matrix whose entries are polynomials. The variables of these polynomials will be the {\em parameters} of the family of {\em simple} polynomials $S$ (these parameters are the coefficients of the linear forms appearing in the decompositions in both the tensor rank and Waring rank settings). Call this symbolic matrix $L(S)$.
\item Next, the {\em symbolic} rank of $L(S)$
(over the field of rational functions in these variables) is bounded by the maximum rank of any {\em evaluation} of this matrix polynomial (this is the only place we use the fact that the field is large enough). By assumption, as these evaluations are all in the image of $L$ on the simple polynomials $S$, this maximum rank is at most $r$, and so is the symbolic rank. 
\item The symbolic rank gives rise to a decomposition $L(S)=KM$ with $M,K$ having dimensions $m\times r$ and $r\times m$ respectively, and their entries are {\em rational functions} in the variables appearing in $L(S)$. We show that with a small loss in the dimension $r$, this affords a much nicer decomposition $L(S)=K'M'$, with dimensions $m\times r'$ and $r' \times m$ respectively, but now the entries of $K',M'$ are {\em polynomial} functions of the variables. Moreover, the polynomials in every column of $K'$ and every row of $M'$ are homogeneous of the same degree. For tensor rank we obtain $r'=r2^d$, and for Waring rank we have $r'=r(d+1)$. 
\item As all entries in matrix $L(S)$ are polynomials of degree $d$, we must have for every $i\in [r']$, that either the $i$'th column of $K'$ or the $i$'th row of $M'$ have degree at most $\lfloor d/2 \rfloor$. The dimension of the space of (vector) coefficients of these vectors of polynomials is an appropriate function $D$ of $n,d$ (which in both cases we care about is about $n^{\lfloor d/2 \rfloor}$). Each such vector of polynomials generates at most $D$ constant vectors of their coefficients.
\item Combining what we have, we see that for every $g\in S$, we have a decomposition  $L(g)=C+R$, where the columns of $C$ are spanned  by at most $r'D$ vectors, and the rows of $R$ are spanned by at most $r'D$ vectors (indeed the total number of these vectors is $r'D$). This gives an upper bound of $r'D$ on the rank of each $L(g)$, which of course is not interesting as we already have an upper bound of $r$ on each.
\item The punchline is obtained by using the linearity of $L$, and the fact that 
 $\hat S$ is the linear span of $S$. Together, these imply that 
{\em every} matrix $L(f)$ with $f\in \hat S$ is also in the linear span of the matrices $\{ L(g) \,:\, g\in S\}$, and so the same decomposition holds for them. Thus, the rank of each $L(f)$ is at most $r'D$, which is a bound on $\mu_L(\hat S)$. Thus, $c(\Delta_0^S) \leq r'D/r$. In the two settings we consider, $D$ is roughly the best known explicit lower bound, and $r'/r$ is a function of $d$  (namely, $d+1$ for Waring rank, and $2^d$ for tensor rank).

\end{itemize}

\subsection{Related Work}

We now mention other attempts to provide barriers to arithmetic circuit lower bounds. 
We also mention rank lower bounds in Boolean complexity, and barriers for them. As will 
be evident, our work is very different than both sets. 

All barrier results we are aware of in arithmetic complexity theory attempt to find analogs of the {\em natural proof} barrier in Boolean circuit complexity of Razborov and Rudich~\cite{RR94}. Roughly, a lower bound technique is {\em natural} if it satisfies three properties: {usefulness, constructively, largeness} which we will not need to define. They show how many Boolean circuit lower bound techniques satisfy these properties. Now crucially, the barrier results for natural proofs in the Boolean setting are {\em conditional}:  they hold under a computational assumption on the existence of efficient pseudorandom generators. In this setting, this assumption is widely believed, and is known to follow from e.g. the existence of exponentially hard one-way functions (one which the world relies for cryptography and e-commerce).

In several works, starting with~\cite{AD08,G15}, and following with the 
recent~\cite{FSV17, GKSS17}, it was understood that an analogous framework with the same three properties is simple to describe (replacing the representation of Boolean functions by their  truth tables by the representation of low-degree multivariate polynomials by their list of coefficients). And indeed, it captures essentially all arithmetic lower bounds known. Unfortunately, the main difference from the Boolean setting is the non-existence of an analogous pseudo-randomness theory, and a believable complexity assumption. Several suggestions for such an assumption were made in the works above, and as articulated in~\cite{FSV17, GKSS17}, they all take the form of the existence of {\em succinct} hitting sets for small arithmetic circuits (indeed, such existence is {\em equivalent} to a barrier result). This assumption is related to PIT (polynomial identity testing) and GCT (geometric complexity theory), but the confidence in it is still shaky (initial work in \cite{FSV17} shows succinct hitting sets against extremely weak models of arithmetic circuits). But regardless how believable this assumption is, note that this barrier is again, conditional!

As mentioned earlier, our barrier results are completely unconditional, and moreover require no constructivity from the lower bound proof (thus capturing methods which are not strictly natural in the sense above). On the other hand, our framework of rank methods capture only a large subset, but certainly not all of the known lower bound techniques.

It is interesting that rank methods were used not only in arithmetic complexity, but also in Boolean complexity.  While not directly related to our arithmetic setting, we mention where it was used, and which barriers were studied. First, Razborov has used the rank of matrices in an essential way for his lower bound on $AC^0[2]$ (although an elegant route around it was soon after devised by Smolensky~\cite{S93}). In another work, Razborov~\cite{Razb99} has shown how rank methods can be used to prove superpolynomial  lower bounds on {\em monotone} Boolean formulas. His methods were recently beautifully  extended to other monotone variants of other models including span programs and comparator  circuits in~\cite{RPRC16}. The potential of such methods to proving {\em non-monotone} lower bounds for Boolean  formulas was considered by Razborov~\cite{Rsub}, where he proves a strong barrier result in this Boolean setting. Observing that rank is a {\em submodular} function, he 
presents a barrier for any submodular progress measure on Boolean formulae: {\em no 
such method can prove a super-linear lower bound!}. His barrier was recently made more explicit in~\cite{Pot16}.

\subsection{Organization}

In Section~\ref{sec:prelim} we establish the notation that will be used 
throughout the paper and provide some lemmas which we will need in the 
later sections. In Section~\ref{sec:matrix-decomp}, we
establish the main technical content of our paper: we define three notions of 
matrix decomposition and relate these new definitions to commutative rank. 
In Section~\ref{sec:rank-bounds}, we apply the
new decompostions from Section~\ref{sec:matrix-decomp} to obtain the main 
results of the paper, which are the limitations of the rank techniques.
In Section~\ref{sec:approach}, we raise the question of what lower bounds can
still be proved using rank methods, and we propose an approach (\emph{using rank methods}) 
for proving better lower bounds for (non-homogeneous) depth-3 formulas. 
Finally, in Section~\ref{sec:conclusion} we conclude the paper and present some open 
questions and future directions of this work.

\section{Preliminaries}\label{sec:prelim}

In this section, we establish the notation which will be used throughout the
paper and some important background which we shall need to prove our claims
in the next sections.

\subsection{General Facts and Notations}

For simplicity of exposition, we will work over a field $\F$ which is algebraically closed and
of characteristic zero, even though our results also hold over infinite fields 
which need not be algebraically closed.\footnote{In general, we only need a field with characteristic
polynomial in the number of variables, the degree of the polynomials and the dimension of matrices being studied. 
We cannot work over field extensions, as we need to use Lemma~\ref{lem:generic-rank} over the base field.} 
From now on we will use boldface to denote a vector of 
variables or of field elements. For instance, $\bx = (x_1, \ldots, x_n)$ is the vector of 
variables $x_1, \ldots, x_n$ and $\ba = (a_1, \ldots, a_n) \in \F^n$ is a vector of elements 
$a_1, \ldots, a_n$ from the field $\F$.  

For any vector of non-negative integers $\ba \in \N^{n}$ and a vector of $n$ variables $\bx$, 
we define $\ba! = \dst \prod_{i=1}^n a_i!$ and 
$\bx^\ba = \dst \frac{1}{\ba!} \cdot \prod_{i=1}^n x_i^{a_i}$. Since the monomials $\bx^\ba$,
$\ba \in \N^n$, form a linear basis for the ring of polynomials $\F[\bx]$, we can write any
polynomial $f(\bx) \in \F[\bx]$ as 
$$ f(\bx) = \sum_{\ba \in \N^n} \alpha_\ba \bx^\ba.  $$
We will denote the coefficients of the polynomial $f(\bx)$ by $\coeff_\ba(f(\bx)) = \alpha_\ba$.

We denote the partial derivative 
$\D_{\ba} = \D_{x_1}^{a_1}\D_{x_2}^{a_2} \cdots \D_{x_n}^{a_n}$. Hence, if we take 
partial derivative $\D_\ba$ of monomial $\bx^{\ba+\bb}$, we get 
$$ \D_\ba \bx^{\ba+\bb} = \bx^\bb. $$ 

The degree of a polynomial $f(\bx) \in \F[\bx]$ with respect to a variable $x_i$, denoted by
$\deg_i(f(\bx))$ is the maximum degree of $x_i$ in a nonzero monomial of $f(\bx)$. If 
$\deg_i(f(\bx)) \leq 1$ for every variable $x_i$, we say that the polynomial $f(\bx)$ is a
{\em multilinear} polynomial. Moreover, if $f(\bx)$ is multilinear and the variables in 
$\bx$ can be partitioned into sets $\bx_1, \ldots, \bx_d$ such that each monomial from $f(\bx)$ 
has at most one variable from each of the sets $\bx_i$, we say that $f(\bx)$ is a 
{\em set-multilinear} polynomial. 

\begin{definition}[Homogeneous Components]
	For a polynomial $f(\bx)$, denote its homogeneous part of degree $t$ by $H_t[f(\bx)]$.
	Additionally, define 
	$$H_{\leq t}[f(\bx)] = \dst\sum_{i=0}^t H_i[f(\bx)],$$ 
	that is, $H_{\leq t}[f]$ is the sum of the homogeneous components of $f(\bx)$ up to degree $t$.
	We can extend this definition to matrices of polynomials in the natural way. Namely, if 
	$\bf(\bx)$ is a matrix of polynomials of the form $(f_{ij}(\bx))_{i,j}$, we define
	$H_t[\bf(\bx)] = (H_t[f_{ij}(\bx)])_{i,j}$, that is, $H_t[\bf(\bx)]$ is the
	matrix given by the homogeneous components of degree $t$ of each entry of $\bf(\bx)$.
\end{definition} 

\begin{definition}[Homogeneous Set Multilinear Components]
	Let $\bx = (\bx_1, \ldots, \bx_d)$ be a set of variables, partitioned into $d$ sets of variables 
	$\bx_1, \ldots, \bx_d$. For a polynomial $f(\bx)$ of degree $d$, let 
	$H^{SM}_S[f(\bx)]$ denote its homogeneous 
	set-multilinear part corresponding to subpartition $S \subseteq [d]$. That is, $H^{SM}_S[f(\bx)]$
	consists of the sum of all monomials (with the appropriate coefficients) of $f(\bx)$ of degree 
	exactly $|S|$ which are set-multilinear with respect to the partition $(\bx_i)_{i \in S}$.
\end{definition} 

\begin{example}
	Let $\bx_1 = (x_{11}, x_{12})$ and $\bx_2 = (x_{21}, x_{22})$ be the variable partition of $\bx = (\bx_1, \bx_2)$. 
	If $f(\bx) = x_{11}^2x_{12} - 3x_{11}x_{21} + 2x_{12}x_{21} - x_{22}^2 + x_{11} - x_{12} + 4x_{21}$, we have that
	$$ H^{SM}_{\{1\}}[f(\bx)] = x_{11} - x_{12},  $$
	whereas
	$$ H^{SM}_{\{1, 2\}}[f(\bz)] = - 3x_{11}x_{21} + 2x_{12}x_{21}. $$
\end{example}

The following lemma tells us that any nonzero polynomial cannot vanish on a large portion
of any sufficiently large grid.

\begin{lemma}[Schwartz-Zippel-DeMillo-Lipton~\cite{S80, Z79, DL78}]
\label{lem:schwartz-zippel}
	Let $\F$ be any field such that $|\F| > d$ and let $S \subseteq \F$ be such that
	$|S| > d$. If $p(\bx) \in \F[\bx]$ is a nonzero polynomial 
	of degree $d$, then
	$$ \Pr_{\ba \in S^n}[p(\ba) = 0] \leq \frac{d}{|S|}. $$
\end{lemma}

\subsection{Matrix Spaces}

In this section, we introduce the concept of matrix spaces and establish some of their 
important properties which we will use in the next sections. We begin by establishing some
notations for matrices and tensors.

If $V$ is a vector space of dimension $n$ over a field $\F$, we can identify $V = \F^n$. 
In this case, we denote the $d^{th}$ tensor power of $V$ by $\ten_{n, d}(\F) = V^{\otimes d}$.
We denote the space of $n \times n$ matrices $V^{\otimes 2}$ by $\mat_n(\F) = \ten_{n,2}(\F)$.
Sometimes we will abuse notation and write $\mat_n(R)$ for the ring of matrices whose entries
take value over a ring $R$.

A tensor $T \in \ten_{n,d}(\F)$ is a rank-1 tensor if it can be written in the form 
$T = \bv_1 \otimes \cdots \otimes \bv_d$, where each $\bv_i \in \F^n$. Given any tensor
$T \in \ten_{n,d}(\F)$, its rank over $\F$ (denoted by $\rk_\F(T)$) is the minimum number $r$ 
of rank-1 tensors $T_1, \ldots, T_r$ such that $T = T_1 + \cdots + T_r$. Whenever the base field is
clear from context, we will denote $\rk_\F(T)$ simply by $\rk(T)$. 

If $M_1, \ldots, M_k$ are matrices in $\mat_m(\F)$ and $x_1, \ldots, x_k$ are commuting variables, we denote 
$\rk_{\F(x_1, \ldots, x_k)}(\sum_{i=1}^k x_i M_i)$ the {\em symbolic rank} of the matrix $\sum_{i=1}^k x_i M_i$.

\begin{definition}[Rank of a Set of Matrices]
	If $\cM \subset \mat_m(\F)$ is a set of $m \times m$ matrices over the field $\F$,
	define 
	$$\rk(\cM) = \max_{M \in \cM} \rk(M).$$ 
	That is, the rank of the set $\cM$ is given by the maximum rank (over $\F$) among its elements.
\end{definition}

The symbolic rank is important as it characterizes the rank of a linear space of matrices, as
seen in the following proposition.

\begin{proposition} Let $\cM \subseteq \mat_m(\F)$ be a space of matrices. 
If $M_1,\ldots M_m$ is a basis for $\cM$ and $x_1,x_2\ldots x_m$ 
are variables then 
$$ \rk(\cM)= \rk_{\F(x_1,\ldots x_{m})} \left(\sum_{i=1}^m x_i M_i \right). $$
\end{proposition}

The propostion above, together with Lemma~\ref{lem:schwartz-zippel}, imply the
following lemma:

\begin{lemma}[Rank Upper Bound on Polynomial Matrices]\label{lem:generic-rank}
	Let $\bx = (x_1, \ldots, x_n)$. If $M(\bx) \in \mat_m(\F[\bx])$ 
	is a matrix
	such that $\rk_\F(M(\ba)) \leq r$ for all $\ba \in \F^n$,
	then $\rk_{\F(\bx)}(M(\bx)) \leq r$.
\end{lemma}

The following proposition shows one way in which a linear space of matrices is
of low rank. This decomposition and its variants will be very useful to us throughout 
the paper.

\begin{proposition}\label{prop:rank-bd-tensor}
	Let $\cM \subset \mat_m(\F)$ be a vector space of matrices such
	that $\cM = \spn(U \otimes V)$, where $U, V \subset \F^m$
	are vector spaces of dimensions $r$ and $s$, respectively. Then,
	$$ \rk(\cM) = \min(r, s). $$
\end{proposition}

\begin{proof}
	W.l.o.g., assume that $r \leq s$. Let $\bu_1, \ldots, \bu_r \in U$ be a basis for the 
	space $U$. As $\cM = \spn(U \otimes V)$, we have that any 
	$M \in \cM$ can be written in the form 
	$$ M = \sum_{i=1}^r \bu_i \otimes \bv_i,  \text{ where } \bv_i \in V, \text{ for } i \in [r].$$
	Hence, $\rk(M) \leq r = \min(r,s)$, for any $M \in \cM$. As 
	$\dst \rk(\cM) = \max_{M \in \cM} \rk(M)$, we obtain that $\rk(\cM) \leq \min(r, s)$, as
	we wanted.
\end{proof}

\subsection{Coefficient Spaces and Their Properties}

As we saw in Section~\ref{sec:proof-overview}, linear spaces of matrices may possess
special structure if they are generated by the coefficients of a matrix of polynomials. 
This observation, together with the definition below, are crucial in obtaining upper bounds
for the rank techniques which we study.

\begin{definition}[Coefficient Space]{\label{def:coef}}
	Let $M(\bx) \in \F[\bx]^{m \times k}$ be a symbolic matrix of polynomials. 
	Considering the monomial basis
	$\{ \bx^\be \}_{\be \in \N^n}$ for the space $\F[\bx]$, we can write 
	$M(\bx) = \dst\sum_{\be \in \N^n} M_\be \cdot \bx^\be$, where each 
	$M_\be \in \F^{m \times k}$ is a matrix of field elements. We define the 
	\emph{coefficient space} of $M(\bx)$, denoted
	by $\cC(M(\bx))$, as the vector space spanned by the vectors $M_\be$. That is,
	$$ \cC(M(\bx)) = \spn\{ M_\be \mid \be \in \N^n \}. $$
	Note that $\cC(M(\bx)) \subseteq \F^{m \times k}$.
\end{definition}

Having the definition above, we proceed to show some nice properties of the
coefficient space of a matrix of polynomials.

\begin{proposition}
	Let $\bx = (x_1, \ldots, x_n)$. If $\bf(\bx)\in \F[\bx]^m$ is a vector of homogeneous 
	polynomials of degree $d$, then 
	$$ \dim(\cC(\bf(\bx))) \leq  \binom{n+d-1}{n-1}. $$
\end{proposition}

By using the proposition above and Propostion~\ref{prop:rank-bd-tensor}, we have the following
corollary:

\begin{corollary}\label{cor:dim-hom-ten}
	Let $\bx = (x_1, \ldots, x_n)$. If $\bf(\bx), \bg(\bx) \in \F[\bx]^m$ are vectors of 
	homogeneous polynomials of degree $d_f$ and $d_g$, respectively, then we have:
	$$\rk(\cC(\bf(\bx) \otimes \bg(\bx))) \leq 
	\min\left\{ \binom{n+d_f-1}{n-1}, \binom{n+d_g-1}{n-1} \right\}.$$
\end{corollary}

The bound above only requires the vectors of polynomials to be homogeneous. 
If these vectors possess more structure, we can obtain better bounds on the rank
of the coefficient space above. 
As we will soon see, if the vectors of polynomials $\bf(\bx), \bg(\bx)$ are vectors of 
set-multilinear polynomials, the two statements below yield a better bound.

\begin{proposition}
	Let $\bx = (\bx_1,\ldots, \bx_d)$ be a set of $nd$ variables, partitioned into $d$ 
	sets of $n$ variables each, denoted by $\bx_i$. If $\bf(\bx)\in \F[\bx]^m$ is a 
	vector of homogeneous and set-multilinear polynomials 
	of degree $d$, with respect to the partition $\bx = (\bx_1,\ldots, \bx_d)$, then 
	$$ \dim(\cC(\bf(\bx))) \leq  n^d. $$
\end{proposition}

By using this new proposition and Propostion~\ref{prop:rank-bd-tensor}, we have the following
corollary:

\begin{corollary}\label{cor:dim-hom-sm-ten}
	Let $\bx = (\bx_1,\ldots, \bx_d)$ be a set of $nd$ variables, partitioned into $d$ 
	sets of $n$ variables each,
	denoted by $\bx_i$. Additionally, let $S_f \sqcup S_g = [d]$ be a partition of the set 
	$[d]$ such that 
	$|S_f| = d_f$ and $|S_g| = d_g$. 
	If $\bf(\bx), \bg(\bx) \in \F[\bx]^m$ are vectors of homogeneous set-multilinear polynomials, 
	where $\bf(\bx)$ is partitioned with respect to the variables $(\bx_i)_{i \in S_f}$ 
	and $\bg(\bx)$ is partitioned with respect to the variables $(\bx_i)_{i \in S_g}$, then we have:
	$$\rk(\cC(\bf(\bx) \otimes \bg(\bx))) \leq \min\left\{ n^{d_f}, n^{d_g} \right\}.$$
\end{corollary}

It is important to observe here that the bound of Corollary~\ref{cor:dim-hom-sm-ten} is 
better than the bound obtained in Corollary~\ref{cor:dim-hom-ten}. To see this, notice that
the number of variables in the setting of Corollary~\ref{cor:dim-hom-sm-ten} is $nd$, and
the degrees of the vectors of polynomials $\bf(\bx)$ and $\bg(\bx)$ are $d_f, d_g$, 
respectively. By using the bounds of Corollary~\ref{cor:dim-hom-ten} we would obtain 
an upper bound of 
$$ \min\left\{  \binom{nd+d_f-1}{nd-1}, \binom{nd+d_g-1}{nd-1} \right\}, $$
and thus weaker than the bound obtained in Corollary~\ref{cor:dim-hom-sm-ten}.

\section{Restricted Forms of Symbolic Matrix Rank Decompositions}\label{sec:matrix-decomp}

If some matrix $M$ over a field $\F$ has rank $r$, then we can write $M$ 
as sum of $r$ matrices $M = M_1 + \ldots + M_r$, where each $M_i$ 
is a rank one matrix over $\F$, and thus can be written as $M_i = \bu_i \otimes \bv_i$, where
$\bu_i, \bv_i$ are vectors over $\F$. In this section we would like to discus what happens when 
we impose additional conditions on the matrix $M$ and on the rank one matrices $M_i$. 

For instance, let $M(\bx) \in \mat_m(\F[\bx])$ be a matrix of homogeneous polynomials of 
degree $d$ such that $\rk_{\F(\bx)}(M) = r$. We want to know 
the minimal $r'$ such that $M(\bx)$ can be written as sum of $r'$ matrices $M_i(\bx)$ of 
rank one, where each $M_i(\bx)$ decomposes as $\bu_i(\bx) \otimes \bv_i(\bx)$ for
$\bu_i(\bx), \bv_i(\bx) \in \F[\bx]^m$ being vectors of homogeneous polynomials. Notice
that this decomposition imposes the condition that the vectors $\bu_i(\bx), \bv_i(\bx)$
be vectors of polynomials, whereas in the general rank decomposition these vectors could
be vectors of rational functions, that is, elements of $\F(\bx)^m$.

In this section, we define some non-standard notions of rank, along with some properties 
which will be useful to us in the next sections. We begin with the definition of homogeneous rank.

\subsection{Homogeneous Rank}

In this section, we define homogeneous rank and then show some interesting properties of 
such decomposition.

\begin{definition}[Homogeneous Rank]
	Let $M(\bx) \in \mat_m(\F[\bx])$ be a matrix of homogeneous polynomials of degree $d$.
	The homogeneous rank of $M(\bx)$, denoted by $\hrk(M(\bx))$ is the minimum $r$ such that 
	$$ M(\bx) = \sum_{i=1}^r \bu_i(\bx) \otimes \bv_i(\bx), $$ 
	where each $\bu_i(\bx), \bv_i(\bx) \in \F[\bx]^m$ is a vector whose entries are 
	homogeneous polynomials of the same degree ($d_{u_i}$ and $d_{v_i}$, respectively). 
\end{definition}

Let $M(\bx) \in \mat_m(\F[\bx])$ be a matrix whose entries are homogeneous polynomials of degree $d$. 
The following lemma shows that if $\rk(M(\bx)) = r$, then it can be written as the homogeneous 
component of degree $d$ of a sum of $r$ rank one matrices with polynomial entries.

\begin{lemma}[Symbolic Matrix Decomposition Lemma]\label{lem:symbolic-decomposition} 
	Let $M(\bx) \in \mat_m(\F[\bx])$ be a matrix of homogeneous polynomials of degree $d$. 
	If $\rk_{\F(\bx)}(M(\bx)) = r$ then there are vectors
	$\bf_1(\bx), \ldots, \bf_r(\bx) \in \F[\bx]^m$ and 
	$\bg_1(\bx), \ldots, \bg_r(\bx) \in \F[\bx]^m$ such that  
	\begin{equation*}
		M(\bx) = \sum_{i=1}^r H_d[\bf_i(\bx) \otimes \bg_i(\bx)].
	\end{equation*}
\end{lemma}

\begin{proof}
	Since $\rk_{\F(\bx)}(M(\bx)) = r$, there exist $r$ pairs of
	vectors of polynomials $\bp_i(\bx), \bq_i(\bx) \in \F[\bx]^m$ and nonzero polynomials 
	$t_i(\bx) \in \F[\bx]$ such that 
	\begin{equation*}
		M(\bx) = \sum_{i=1}^r \dfrac{1}{t_i(\bx)} \bp_i(\bx) \otimes \bq_i(\bx).
	\end{equation*}
	
	Since $t_i(\bx)$ are nonzero polynomials for all $i \in [r]$, the polynomial given by $Q(\bx) = \dst\prod_{i=1}^r t_i(\bx)$
	is a nonzero polynomial. By $\char(\F) = 0$ and Lemma~\ref{lem:schwartz-zippel}, there exists 
	$\ba \in \F^n$ such that
	$Q(\ba) \neq 0$. In particular, this implies that we can write $t_i(\bx + \ba) = b_i \cdot (1 - \hat{t}_i(\bx))$,
	where $b_i \in \F$ are nonzero field elements and 
	$\hat{t}_i(\bx)$ are polynomials such that $\hat{t}_i(\bo) = 0$. Namely, the constant
	terms of $\hat{t}_i(\bx)$ are zero, for all $i \in [r]$.  
	
	Writing $\widehat{\bp}_i(\bx) = \bp_i(\bx + \ba)$, 
	$\widehat{\bq}_i(\bx) = \bq_i(\bx + \ba)$, and from 
	the power series expansion of $1/(1 - x)$, it follows that 
	
	\begin{align*}
		M(\bx + \ba) 
		&= \sum_{i=1}^r \dfrac{1}{t_i(\bx+\ba)} \widehat{\bp}_i(\bx)  \otimes \widehat{\bq}_i(\bx) \\
		&= \sum_{i=1}^r 	\dfrac{1}{b_i \cdot (1- \hat{t}_i(\bx))}	\widehat{\bp}_i(\bx)  \otimes \widehat{\bq}_i(\bx) \\
		&= \sum_{i=1}^r \dfrac{1}{b_i} \left[	\widehat{\bp}_i(\bx)  \otimes \widehat{\bq}_i(\bx) \right] 
		\cdot \left( \sum_{j=0}^\infty \hat{t}_i(\bx)^j \right). 
	\end{align*}		
	As $M(\bx + \ba)$ is a matrix of polynomials of degree no larger than $d$, the equality above becomes:
	\begin{align*}
		M(\bx + \ba) &= H_{\leq d}[M(\bx + \ba)] \\
		&= H_{\leq d} \left\{ \sum_{i=1}^r \dfrac{1}{b_i} \left[	\widehat{\bp}_i(\bx)  \otimes \widehat{\bq}_i(\bx) \right] 
		\cdot \left( \sum_{j=0}^\infty \hat{t}_i(\bx)^j \right) \right\} \\ 
		&= H_{\leq d} \left\{ \sum_{i=1}^r \dfrac{1}{b_i} \left[	\widehat{\bp}_i(\bx)  \otimes \widehat{\bq}_i(\bx) \right] 
		\cdot \left( \sum_{j=0}^d \hat{t}_i(\bx)^j \right) \right\} \\
		&= \sum_{i=1}^r H_{\leq d}[ \wt{\bp}_i(\bx) \otimes \wt{\bq}_i(\bx) ], 
	\end{align*}	
	where $\wt{\bp}_i(\bx) = \dfrac{1}{b_i} \widehat{\bp}_i(\bx)$
	and $\dst \wt{\bq}_i(\bx) = \widehat{\bq}_i(\bx) \cdot \left( \sum_{j=0}^d \hat{t}_i(\bx)^j \right)$.
	
	Moreover, from homogeneity of $M(\bx)$, we have $M(\bx) = H_d[M(\bx + \ba)]$, which implies
	\begin{align*}
		M(\bx) = H_d[M(\bx + \ba)] 
		= \sum_{i=1}^r H_d[ \wt{\bp}_i(\bx) \otimes \wt{\bq}_i(\bx) ].
	\end{align*}			
	Taking $\bf_i(\bx) = \wt{\bp}_i(\bx)$ and $\bg_i(\bx) = \wt{\bq}_i(\bx)$ completes the proof.
\end{proof}

The following lemma uses the decomposition above to prove that, essentially, if a matrix 
whose entries are homogeneous polynomials has small rank then such a matrix also has small homogeneous rank.

\begin{lemma}[Matrix Rank Over Polynomial Rings]\label{lem:poly-rank}
	Let $M(\bx) \in \mat_m(\F[\bx])$ be a matrix with polynomial entries such that
	each entry $M_{ij}(\bx)$ is a homogeneous polynomial of degree $d$. 
	$$ \text{If } \rk_{\F(\bx)}(M(\bx)) \leq r \text{ then } \hrk(M(\bx)) \leq r \cdot (d+1).$$
\end{lemma}

\begin{proof}
	W.l.o.g., we can assume that $\rk_{\F(\bx)}(M(\bx)) = r$.  From 
	Lemma~\ref{lem:symbolic-decomposition}, there exist vectors of polynomials 
	$\bp_1(\bx), \bq_1(\bx), \ldots, \bp_r(\bx), \bq_r(\bx) \in \F[\bx]^m$ such that
	
	\begin{equation}\label{eq:hom-d-regular-1}
		M(\bx) = \sum_{i=1}^r H_d[\bp_i(\bx) \otimes \bq_i(\bx)].
	\end{equation}

	Decomposing equality~\eqref{eq:hom-d-regular-1} into its homogeneous components, we obtain:
	\begin{align*}
		M(\bx)
		= \sum_{i=1}^r H_d[ \bp_i(\bx) \otimes \bq_i(\bx) ]
		= \sum_{i=1}^r \sum_{k=0}^d  H_k[\bp_i(\bx)] \otimes H_{d-k}[\bq_i(\bx)].
	\end{align*}	
	The last line of the equality above gives us the decomposition of $M(\bx)$ into 
	$\leq r \cdot (d+1)$ rank-1 polynomial matrices.
\end{proof}

%
%

\subsection{Set-Multilinear Rank}

While the decomposition of a matrix with polynomial entries into homogeneous rank one matrices is an
important one, it may not be the best decomposition if the original matrix has additional
structure. An example is the decomposition of a 
set-multilinear polynomial matrix into set-multilinear rank one matrices. To that extent, 
we need the following notion:
 
\begin{definition}[Set-Multilinear Rank]
	Let $\bx = (\bx_1, \ldots, \bx_d)$ be a set of variables, partitioned into sets of variables
	$\bx_i$, and $M(\bx) \in \mat_m(\F[\bx])$ be a matrix with polynomial entries such that
	each entry $M_{ij}(\bx)$ is a homogeneous set-multilinear polynomial of degree $d$,
	where the partition is given by $\bx$. \\
	The {\em set-multilinear rank} of $M(\bx)$, denoted by $\smrk(M(\bx))$, is the 
	smallest integer $r$
	for which there exist $r$ pairs of vectors $\bf_i(\bx), \bg_i(\bx) \in \F[\bx]^m$ such that
	\begin{equation}\label{eq:sm-decomposition}
		M(\bx) = \sum_{i=1}^{r} \bf_i(\bx) \otimes \bg_i(\bx), 
	\end{equation}	  
	where: 
	\begin{itemize} 
		\item $\bf_i(\bx)$ and $\bg_i(\bx)$ are homogeneous vectors of
		set-multilinear polynomials, 
		\item for each $i \in [r]$, there exists a partition $S_f^i \sqcup S_g^i = [d]$
		of  the set $[d]$ such that $\bf_i(\bx)$ is set-multilinear with respect to the variables
		$(\bx_j)_{j \in S_f^i}$ and $\bg_i(\bx)$ is set-multilinear with respect to the variables 
		$(\bx_j)_{j \in S_g^i}$. 
	\end{itemize}
	In particular, $\deg(\bf_i(\bx)) + \deg(\bg_i(\bx)) = d$. 
\end{definition}

With this concept of set-multilinear decomposition, we obtain the following relationship
between the symbolic rank and the set-multilinear rank of a set-multilinear polynomial matrix. 

\begin{lemma}[Set-Multilinear Rank of Polynomial Matrices]\label{lem:poly-rank-set-multilinear}
	Let $M(\bx) \in \mat_m(\F[\bx])$ be a set-multilinear matrix of degree $d$,
	with partition $\bx = (\bx_1, \ldots, \bx_d)$. 
	$$ \text{If } \rk_{\F(\bx)}(M(\bx)) \leq r \text{ then } \smrk(M(\bx)) \leq r \cdot 2^d. $$
\end{lemma}

\begin{proof}
		W.l.o.g., we can assume that $\rk_{\F(\bx)}(M(\bx)) = r$.  From 
		Lemma~\ref{lem:symbolic-decomposition}, there exist vectors of polynomials 
		$\bp_1(\bx), \bq_1(\bx), \ldots, \bp_r(\bx), \bq_r(\bx) \in \F[\bx]^m$ such that
	
	\begin{equation}\label{eq:hom-d-regular-2}
		M(\bx) = \sum_{i=1}^r H_d[\bp_i(\bx) \otimes \bq_i(\bx)].
	\end{equation}	
	
	Decomposing equality~\eqref{eq:hom-d-regular-2} into its homogeneous and set multilinear 
	components, according to the partition $\bx = (\bx_1, \ldots, \bx_d)$ we obtain:
	\begin{align*}
		M(\bx)
		= \sum_{i=1}^r H_{[d]}^{SM}[ \bp_i(\bx) \otimes \bq_i(\bx) ]
		= \sum_{i=1}^r \sum_{S \subseteq [d]}  H_S^{SM}[\bp_i(\bx)] \otimes 
		H_{[d] \setminus S}^{SM}[\bq_i(\bx)].
	\end{align*}	
	The last line of the equality above giving us the decomposition of $M(\bx)$ into 
	$R \leq r \cdot 2^d$ rank-1 polynomial matrices.
\end{proof}

Note that the set-multilinear decomposition of Lemma~\ref{lem:poly-rank-set-multilinear}
is much more stringent than the homogeneous decomposition of Lemma~\ref{lem:poly-rank}, and therefore gives
seemingly worse bounds. However, as we will see in Section~\ref{sec:rank-bounds}, the set-multilinear 
decomposition shown in equation~\eqref{eq:sm-decomposition} will turn out to be better than the 
homogeneous one (when applicable). This is due to the vectors of polynomials $\bf_i(\bx), \bg_i(\bx)$ 
in~\eqref{eq:sm-decomposition} being much simpler than the ones obtained in the homogeneous decomposition.
In particular, the set multilinearity of $\bf_i(\bx), \bg_i(\bx)$ in~\eqref{eq:sm-decomposition} will yield simpler
coefficient spaces, and therefore better rank bounds on the coefficient space of $M(\bx)$.

\section{Rank Bounds}\label{sec:rank-bounds}

In this section, we show how the matrix decomposition techniques developed in 
Section~\ref{sec:matrix-decomp} can be used to establish barriers to rank-based
methods used to prove lower bounds for tensor rank or for Waring rank and constant
depth circuits. We begin with the Waring rank of a homogeneous polynomial, which is 
defined as follows:

\begin{definition}[Waring Rank]
	Given a homogeneous polynomial $f(\bx) \in \F[\bx]$ in $n$ variables of degree $d$,
	its {\em Waring rank}, written $\wrk(f(\bx))$, is the minimum integer
	$r$ such that $f(\bx)$ can be written as a sum of $r$ $d^{th}$ powers of linear forms.
	That is, there exist linear polynomials $\ell_i(\bx) \in \F[\bx]$, for $1 \leq i \leq r$
	such that 
	$$ f(\bx) = \sum_{i=1}^r \ell_i(\bx)^d. $$ 
\end{definition} 
 
\subsection{Barriers to Waring Rank Lower Bounds}

In this section, we establish barriers for rank-based methods used to prove lower bounds
for the Waring rank of {\em any} polynomial. More precisely, we show
that any rank method which can be cast as a linear map $L : \F[\bx] \to \mat_m(\F)$ such that 
$\rk(L(\ell(\bx)^d)) \leq r$ for all powers of linear forms has the property that 
$\rk(L(f(\bx))) \leq r(d+1) \cdot \binom{n+\lfloor d/2 \rfloor}{n}$ for any polynomial 
$f(\bx)$. This, in turn, implies that such a rank technique cannot yield better lower bounds than 
$$\wrk(f(\bx)) \geq (d+1) \cdot \binom{n+\lfloor d/2 \rfloor}{n}$$ 
for any polynomial $f(\bx)$, as will be shown in Corollary~\ref{cor:waring-rank}.

Currently, the best known lower bounds for the Waring rank of an explicit polynomial in $n$
variables of degree $d$ is 
$$ \binom{n+ \lfloor d/2 \rfloor - 1}{n} + \lfloor n/2 \rfloor -1, $$
as shown in~\cite{GL17}, through a rank-based method. Note that the lower bound 
of~\cite{GL17} gets very close to the barrier that we prove for Waring rank. 

On the other hand, the Alexander-Hirschowitz theorem~\cite{AH95} tells us 
that a random homogeneous polynomial $f(\bx)$ on $n$ variables of degree $d$ has Waring rank 
$$ \wrk(f(\bx)) = \left\lceil \frac{1}{n} \cdot \binom{n+d-1}{n-1} \right\rceil. $$ 
Thus, this work shows that new 
techniques are needed to overcome the gap between the current lower bounds for explicit 
polynomials and the correct lower bound for the Waring rank of random polynomials.

\begin{theorem}[Waring Rank Upper Bounds]\label{thm:waring-rank}
	Let $m,n \in \N$ be positive integers and $L : \F[\bx] \to \mat_m(\F)$ be a linear
	map. If for every affine form $\ell(\bx)$ we have that $\rk(L(\ell^d)) \leq r$, then it holds that 
	$$\rk(L(f(\bx))) \leq r(d+1) \cdot \binom{n+ \lfloor d/2 \rfloor}{n} $$
	for any polynomial $f(\bx)$ of degree at most $d$. 
\end{theorem}

\begin{proof} 
	Since $L$ is a linear map and $\{ \bx^\be \}_{\be \in \N^n}$ is a basis for the vector 
	space of polynomials 
	$\F[\bx]$, we can define $L$ by first defining the matrices $M_\be = L(\bx^\be)$ and 
	extend this definition to all polynomials by linearity.
	Now instead of looking on $L$ as a linear mapping we can look on it as a matrix of 
	polynomials in the following way: 
	For any affine form given by $\ell(\bx) = y_0 + \dst\sum_{i=1}^n y_i x_i$, we have that 
	$$ \ell(\bx)^d = d! \cdot 
	\dst\sum_{\substack{(e_0, \be) \in \N^{n+1} \\ \| (e_0, \be) \|_1 = d}} 
	\frac{\be!}{e_0!} \cdot (y_0^{e_0} \cdot \by^\be) \cdot \bx^\be $$	
	hence, the image of $\ell(\bx)$ under the map $L$ is given by
	$$M(y_0, \by) = L(\ell(\bx)^d) = d! \cdot
	\dst\sum_{\substack{(e_0, \be) \in \N^{n+1} \\ \| (e_0, \be) \|_1 = d}} 
	\frac{\be!}{e_0!} \cdot (y_0^{e_0} \cdot \by^\be) \cdot M_\be$$
	and thus $M(y_0, \by) = L(\ell(\bx)^d)$ is a polynomial matrix (over the variables 
	$y_0, y_1, \ldots, y_n$) where
	each entry is a homogeneous polynomial of degree $d$. By assumption, we have that 
	$\rk(M(a_0, \ba)) \leq r$ for
	any assignment $(a_0, \ba) \in \F^{n+1}$. Therefore, Lemma~\ref{lem:generic-rank} implies that 
	$\rk_{\F(y_0, \by)}(M(y_0, \by)) \leq r$.
	
	In this case, the conditions of Lemma~\ref{lem:poly-rank}
	apply and therefore there exist $R \leq r(d+1)$ vectors of polynomials 
	$\bf_i(y_0, \by), \bg_i(y_0, \by) \in \F[y_0, \by]^m$ such that
	\begin{equation*}
		M(y_0, \by) = \sum_{i=1}^R \bf_i(y_0, \by) \otimes \bg_i(y_0, \by)
	\end{equation*}
	Moreover, for all $i \in [R]$, $\bf_i(y_0, \by), \bg_i(y_0, \by)$ are vectors of 
	homogeneous polynomials such that 
	$\deg(\bf_i) + \deg(\bg_i) \leq d$. Hence, 
	$\min(\deg(\bf_i), \deg(\bg_i)) \leq \lfloor d/2 \rfloor$, for each 
	$i \in [R]$. This bound on the minimum degree, combined with 
	Corollary~\ref{cor:dim-hom-ten}, yields 
	$$\rk(\cC(\bf_i(y_0, \by) \otimes \bg_i(y_0, \by))) \leq 
	\binom{n+\lfloor d/2 \rfloor}{n}.$$
	
	As $\rk(\cC(M(y_0, \by))) \leq 
	\dst\sum_{i=1}^R \rk(\cC(\bf_i(y_0, \by) \otimes \bg_i(y_0, \by)))$,
	we obtain
	$$ \rk(\cC(M(y_0, \by))) \leq R \cdot \binom{n+\lfloor d/2 \rfloor}{n}. $$
	
	Now, to finish the proof, it is enough to show that $L(f(\bx)) \in \cC(M(y_0, \by))$ for any
	polynomial $f(\bx)$ of degree at most $d$. To do this, first notice that for any affine form 
	$\ell(\bx) = a_0 + \dst\sum_{i=1}^n a_i x_i$, where $(a_0, \ba) \in \F^{n+1}$, 
	we have that $L(\ell(\bx)^d) \in \cC(M(y_0, \by))$.
	By linearity of $L$ and the fact that the set of all affine forms of degree $d$ span 
	the space of all polynomials of degree at most $d$, we have:
	$$ L(f(\bx)) \in \spn\left\{ L(\ell(\bx)^d) \ \middle| \  
	\ell(\bx) = a_0 + \dst\sum_{i=1}^n a_i x_i \text{ and } (a_0, \ba) \in \F^{n+1}  \right\}.$$
	As the span of all matrices $L(\ell(\bx)^d)$ is contained in the space $\cC(M(y_0, \by))$,
	we have that $L(f(\bx)) \in \cC(M(y_0, \by))$, as we wanted. Consequently, we have
	$$ \rk(L(f(\bx))) \leq R \cdot \binom{n+\lfloor d/2 \rfloor}{n}. $$
\end{proof}

The theorem above implies the following upper bound on rank-based techniques.

\begin{corollary}\label{cor:waring-rank}
	Let $m,n \in \N$ be positive integers and $L : \F[\bx] \to \mat_m(\F)$ be a linear
	map. Then, rank 
	methods which use this linear map cannot prove lower bounds better than 
	$$ \wrk(f(\bx)) > (d+1) \cdot \binom{n+ \lfloor d/2 \rfloor}{n} $$ 
	for any polynomial $f(\bx)$ of degree at most $d$. 
\end{corollary}

\begin{proof}
	Let $s = \wrk(f(\bx))$ be the Waring rank of $f(\bx)$. In this case, we have 
	that
	$$ f(\bx) = \sum_{i=1}^s \ell_i(\bx)^d,  $$
	where each $\ell_i(\bx)$ is an affine form. 
	
	Let $r$ be an upper bound
	on the rank of $\ell(\bx)^d$, for any affine form $\ell(\bx)$.	
	By applying the map $L$ to both sides of the
	equation, we obtain
	\begin{align*}
		L(f(\bx)) = \sum_{i=1}^s L(\ell_i(\bx)^d) 
		&\then \rk(L(f(\bx))) \leq  \sum_{i=1}^s \rk(L(\ell_i(\bx)^d)) \leq s \cdot r \\
		&\then \frac{\rk(L(f(\bx)))}{r} \leq s = \wrk(f(\bx)).
	\end{align*}
	By Theorem~\ref{thm:waring-rank}, we have that 
	$\rk(L(f(\bx))) \leq r(d+1) \cdot \binom{n+ \lfloor d/2 \rfloor}{n}.$
\end{proof}

\subsection{Barriers to Tensor Rank Lower Bounds}

In this section, we prove limitations of rank-based methods which yield lower bounds for
the tensor rank of explicit tensors. Analogous to the Waring rank case, we show that
any linear map, denoted here by $L : \ten_{n,d}(\F) \to \mat_m(\F)$, for which 
$\rk(L(\bu_1 \otimes \cdots \otimes \bu_d)) \leq r$ for all rank one tensors has the property
that $\rk(L(T)) \leq r \cdot 2^d \cdot n^{\lfloor d/2 \rfloor}$ for any tensor 
$T \in \ten_{n,d}(\F).$ This in turn, implies that such a technique cannot yield better 
lower bounds than
$$ \rk(T) > 2^d \cdot n^{\lfloor d/2 \rfloor} $$
for any explicit tensor $T \in \ten_{n, d}(\F)$.

To put this matter into perspective, it is very easy to obtain explicit tensors 
$T \in \tennd$ whose tensor rank is lower bounded by $\rk(T) \geq n^{\lfloor d/2 \rfloor}$.
For instance, one can just take a full-rank matrix in $\mat_{n^{\lfloor d/2 \rfloor}}(\F)$.
Nevertheless, despite much work on tensor rank lower bounds, the best lower bounds for the rank
of explicit tensors are still of the form $\Omega(n^{\lfloor d/2 \rfloor})$, as seen in
the works~\cite{BI11, AFT11, L12}.

On the other hand, it is well-known, see for instance~\cite{Lan17}, that a random tensor has rank on the 
order of $\frac{n^{d-1}}{d}$. Thus, our paper shows that rank-based methods for proving tensor
rank lower bounds will not suffice to prove strong tensor lower bounds. We now state the main theorem
of this section. 

\begin{theorem}[Tensor Rank Upper Bounds]\label{thm:tensor-rank}
	Let $m,n \in \N$ be positive integers and $L : \ten_{n, d}(\F) \to \mat_m(\F)$ be a linear
	map such that each rank one tensor $\bu_1 \otimes \cdots \otimes \bu_d$ is mapped into 
	a matrix $L(\bu_1 \otimes \cdots \otimes \bu_d)$ such that 
	$$\rk(L(\bu_1 \otimes \cdots \otimes \bu_d)) \leq r. $$ 
	Then it holds that 
	$$\dst \rk(L(f)) \leq r \cdot 2^d \cdot n^{\lfloor d/2 \rfloor} $$
	for any tensor $f \in \ten_{n, d}(\F)$. 
\end{theorem}

\begin{proof}
	Let $\bx_1 \otimes \cdots \otimes \bx_d$ be a generic rank one tensor, where 
	$\bx_i = (x_{i1}, \ldots, x_{in})$, with
	$x_{ij}$ being variables which take values from $\F$, for all $i \in [d]$. Additionally, let
	$\bx = (\bx_1, \ldots, \bx_d)$, that is, $\bx$ is the set of all variables involved, taking 
	into account the partitions of the variables.
	As the map $L : \ten_{n, d}(\F) \to \mat_m(\F)$ is a linear map, we must have that 
	\begin{equation*}
		L(\bx_1 \otimes \cdots \otimes \bx_d) = 
		\sum_{i_1, i_2, \ldots, i_d = 1}^n A_{i_1, i_2, \ldots, i_d} \prod_{j=1}^d x_{j i_j} 
	\end{equation*}	  
	where each $A_{i_1, i_2, \ldots, i_d} \in \mat_m(\F)$ is a complex $m \times m$ matrix.\footnote{One 
	can see this by looking at
	the standard basis of the space $\ten_{n, d}(\F)$ given by tensoring the standard basis vectors 
	$\be_{i_1} \otimes \cdots \otimes \be_{i_d}$.} Hence, 
	$M(\bx) = L(\bx_1 \otimes \cdots \otimes \bx_d)$ is a matrix with 
	set-multilinear polynomial entries, 
	where each polynomial is set-multilinear over the sets of variables $\bx_1, \ldots, \bx_d$.
	
	By Lemma~\ref{lem:generic-rank} and the assumption that 
	$\rk(L(\bu_1 \otimes \cdots \otimes \bu_d)) \leq r$
	for any multiset of vectors $\bu_i \in \F^n$, we have that 
	$$ \rk_{\F(\bx)}(L(\bx_1 \otimes \cdots \otimes \bx_d)) \leq r. $$ 
	In this case, the conditions of Lemma~\ref{lem:poly-rank-set-multilinear} apply and
	therefore there exist $R \leq r \cdot 2^d$ vectors of homogeneous set-multilinear
	polynomials $\bf_i(\bx), \bg_i(\bx) \in \F[\bx]$ for which
	$$ M(\bx) = \sum_{i=1}^R \bf_i(\bx) \otimes \bg_i(\bx). $$
	Moreover, for all $i \in [R]$, 
	there exists a set $S_i$ such that $\bf_i(\bx)$ is 
	set-multilinear with respect to the partition $(\bx_j)_{j \in S_i}$ and $\bg_i(\bx)$
	is set-multilinear with respect to the partition $(\bx_j)_{j \in [d] \setminus S_i}$.
	Thus, $\deg(\bf_i) + \deg(\bg_i) \leq d$, which implies that 
	$\min(\deg(\bf_i), \deg(\bg_i)) \leq \lfloor d/2 \rfloor$, for each $i \in [R]$.
	This bound on the minimum degree, combined with Corollary~\ref{cor:dim-hom-sm-ten} and the fact that 
	$\bf_i(\bx)$ and $\bg_i(\bx)$ are set-multilinear, yield
	$$ \rk(\cC(\bf_i(\bx) \otimes \bg_i(\bx))) \leq n^{\lfloor d/2 \rfloor}. $$ 
	As $\dst \rk(\cC(M(\bx))) \leq \sum_{i=1}^R \rk(\cC(\bf_i(\bx) \otimes \bg_i(\bx)))$, we have that
	$$ \rk(\cC(M(\bx))) \leq R \cdot  n^{\lfloor d/2 \rfloor}. $$
	To finish the proof, it is enough to show that $L(f) \in \cC(M(\bx))$, for any $f \in \ten_{n, d}(\F)$. 
	
	For any rank one tensor $\bu_1 \otimes \cdots \otimes \bu_d$, we have that 
	$L(\bu_1 \otimes \cdots \otimes \bu_d) \in \cC(M(\bx))$,
	as $L(\bu_1 \otimes \cdots \otimes \bu_d) = M(\bu)$. As any element $f \in \ten_{n, d}(\F)$ can 
	be written as a linear combination of rank one tensors and by linearity of $L$, we have that 
	$$ L(f) \in \spn\left\{ L(\bu_1 \otimes \cdots \bu_d) \mid 
	\bu_1, \ldots, \bu_d \in \F^n  \right\} \subseteq \cC(M(\bx)).$$
	Thus, $L(f) \in \cC(M(\bx))$ and we have that 
	$$ \rk(L(f)) \leq \rk(\cC(M(\bx))) \leq R \cdot  n^{\lfloor d/2 \rfloor},$$
	as we wanted.
\end{proof}

The theorem above implies the following barrier on rank-based techniques.

\begin{corollary}\label{cor:tensor-rank}
	Let $m,n \in \N$ be positive integers and $L : \ten_{n, d}(\F) \to \mat_m(\F)$ be 
	a linear map (i.e., a flattening). Then, any rank methods which use this linear 
	map cannot prove lower bounds better than 
	$$\dst \rk(f) > 2^d \cdot n^{\lfloor d/2 \rfloor} $$
	for any tensor $f \in \ten_{n, d}(\F)$. 
\end{corollary}

\section{Approach to Prove Lower Bounds}\label{sec:approach}

Given our main results on barriers to the rank method, it is very interesting to determine for which circuit classes we {\em can} improve the state of the art in terms of lower bounds.
We suggest the class of depth-3 {\em non-homogeneous} formulas, and translate into this linear-algebraic framework what a lower bound proof will require. Needless to say, we have no idea if such a plan can work.

The works~\cite{K12, GKKS13} show that proving strong enough lower bounds 
(better than $\exp(\tilde{O}(\sqrt{n}))$ even for depth 3 formulas) would lead to a 
separation between $\cVP$ and $\cVNP$, so this is certainly a formidable challenge. 
The works~\cite{GKKS14,KLSS14, FLMS15, KS14, KS15} use rank methods to 
prove matching lower bounds of $\exp(\tilde{O}(\sqrt{n}))$ for homogeneous depth-4 
formulas. However, up to date, we do not know better lower bounds better than cubic for 
(non-homogeneous) depth-3 formulas~\cite{KST16}.

We lay out in the language of {\em rank methods} what it would take to prove stronger bounds. 
We will begin by restating the framework and then provide the exact requirement 
specialized to the case of depth-3 formulas. We end this section with a short informal 
discussion of why we feel that this approach is worth pursuing.

%

\paragraph{Restatement of the rank methods:} let $S \subset \hat{S}$ be a set of 
"simple" polynomials and let $\psi :\F^t \rightarrow \hat{S}$ be a polynomial mapping 
such that image of $\psi$ is $S$. That is, $\F^t$ is the space that parameterizes $S$. Let 
$\by = (y_1, \ldots, y_t)$ and define 
${\cal{L}}_S= \{\ell(\psi): \ell \in \hat{S}^*\} \subset \F[\by]$.  
Then, there exists a rank method which proves a lower bound of $R$ for the circuit class 
$\hat{S}$ if and only if there exists a matrix $L$ with entries from ${\cal{L}}$ such that 
$\rk_{\F(\by)}(L) \leq r$ and $\rk(\cC(L)) \geq rR$ for some $r$.

\paragraph{Instantiation for depth-3 formulas:} A depth-3 formula (not necessarily homogeneous) computing a homogeneous polynomial $f(\bx)$ of degree $d$ can be written as follows:
$$ f(\bx) = \sum_{i=1}^s \alpha_i \prod_{j=1}^D (1 + \ell_{ij}(\bx)), $$
where each $\alpha_i \in \F$, $\ell_{ij}$ is a linear form on the variables $\bx$ and 
$D$ is the maximum degree of a gate computed in the formula. Note that $D$ can be 
much larger than $d$\footnote{And indeed, making for non-homogeneous depth-3 circuits taking $D>d$ can yield exponential savings, as for computing symmetric polynomials.}. The size of
our formula is given by $s + D$, as we do not require all product gates to be the product
of $D$ affine forms. Thus, we want to obtain a lower bound for $s+D$.

Given this decomposition, the space of ``simple'' functions $S_D$ is the set containing 
each polynomial of the form 
$$ H_d\left[\prod_{i=1}^D (1+\ell_i(\bx))\right] = 
\Sym_d(\ell_1(\bx),\ell_2(\bx),\ldots \ell_D(\bx)),$$ 
where each $\ell_i$ is a linear form and $\Sym_d$ is 
the homogeneous multilinear symmetric polynomial.
Letting $\by=(y_{ij})$, where $i \in [D]$ and $j \in [n]$, our polynomial mapping 
$\psi: \F^{Dn} \rightarrow \F[\by][\bx]_d$ is defined as follows:   
\begin{equation}\label{eq:poly-mapping}
	\psi(\by)= H_{2d}\left[\prod_{i=1}^D \left(1+\sum_{j=1}^n(y_{ij} x_j) \right)\right]= 
	\Sym_d\left(\sum_{j=1}^n y_{1,j} x_j,\sum_{i=1}^n y_{2,j} x_j, \ldots, 
	\sum_{i=1}^n y_{D, j} x_j \right),
\end{equation} 
where we can identify $\F[\by][\bx]_d$ with the space $\F[\by]^{\binom{n+d-1}{n}}$
in which the coordinates are indexed by monomials of degree $d$ in $\bx$.
 
We claim that the space of possible entries of our matrix $L$, denoted as above by 
$\cal{L}=\{\ell(\psi) : \ell \in \F[\by][\bx]^*_d \}$, is equal to the space of {\em set 
symmetric multilinear} polynomials of degree $d$ in the variables $\by$, which we 
denote by $SSM(\by)$. Note that any $\ell(\psi)$ is a set multilinear function in
$\by = (\by_1, \ldots, \by_D)$, where  $\by_i=(y_{ij})_{j=1}^n$ and 
$\ell(\psi)$ is also {\em set symmetric} in the sense that 
$\ell(\psi)(\by_1,\ldots \by_D) = \ell(\psi)(\by_{\sigma(1)},\ldots, \by_{\sigma(D)})$ 
for every $\sigma \in S_D$ since $\Sym_{2d}$ (and therefore $\psi$ and $\ell(\psi)$) 
has this property. Hence, we have that $\cL \subseteq SSM(\by)$. 

To see that $SSM(\by) = \cL$ we only need to show that $\dim(\cL) = \dim(SSM(\by))$. 
To do so first we note that $\dim(\cal{L})$ is equal to the dimension of the space 
of homogeneous polynomials of degree $d$ over $\bx$, as we can take the standard linear
functionals corresponding to each coordinate of $\F[\by][\bx]_d$ (and the coefficients
in $\by$ of each coordinate are independent). 
Now, note that the space of homogeneous polynomials is isomorphic to the space of set 
symmetric multilinear polynomials by the following homomorphism.  
$$x_{i_1} x_{i_2}\ldots x_{i_{d}} \mapsto 
\sum_{\sigma\in S_n} \prod_{j=1}^d x_{\sigma(i_j)}.$$ 

With these facts in hand, we can now describe the instantiation in linear algebraic terms:

\begin{enumerate}
	\item Our set of simple polynomials is the set $S_D$ described above.
	\item Our polynomial mapping is the map $\psi$ from equation~\eqref{eq:poly-mapping}.
	\item One can take any linear map $L : \hat{S}_D \to \mat_m(\F)$ whose symbolic
	matrix $L(\psi)$ has only set symmetric multilinear polynomial entries. Let $r$
	be the rank of the symbolic matrix $L(\psi)$.
	\item To prove a lower bound of $R$ on the size of depth 3 formulas, one has to
	show that the rank of the space of matrices given by $L(\hat{S}_D)$ is at least
	$r \cdot R$.
\end{enumerate}

\paragraph{Potential value of this formulation:} First, we note that  {\em any} rank method  proving lower bounds for depth-3 formulas will have this form.  With such structured algebraic condition written explicitly, it could be easier to test different choices of mappings, and hopefully improve the state of the art for
depth 3 formulas. Of course, the main challenges are to try to understand the ranks of 
matrices with such special entries, as well as to try to find structured decompositions
of the symbolic matrices generated by this method.  So, while more concrete and algebraic 
than studying depth-3 formulas in general, these questions still seem formidable. 
Our hope that such a study will lead to 
either new lower bounds, or to a new barrier result for this model.

%

\section{Conclusion and Open Problems}\label{sec:conclusion}

In this paper, we prove the first unconditional barrier for a wide class of lower bound
techniques for tensor rank as well as the Waring rank of a polynomial. In particular,
for 3-dimensional tensor rank, we show for the first time that a wide class of 
techniques cannot improve a known linear lower bound (of $2n$) even beyond $8n$.  
Additionally, we provide an explicit instantiation of the rank method for depth-3 circuits, suggesting it will either help prove better
lower bounds, or help develop a barrier for this model that explains the difficulty of proving better lower bounds.  

We now provide a list of interesting directions for further research, both on the
computational side as well as on the mathematical side.

\begin{enumerate}
        \item Expand the set of methods for which {\em unconditional} barrier results be proven in arithmetic complexity theory, beyond the rank methods we study in this paper. In particular, can they be expanded to the use of {\em non-linear} mappings $L$, possibly of low degree?
        
	\item Expand the set of arithmetic models for which barriers can be established for rank methods, beyond the two models studied here.
	
	\item In some sense, rank methods ``flatten'' polynomials of degree $d>2$ into matrices (in $2$ dimensions), in a similar fashion flattening methods in algebraic geometry are used (for very similar purposes). Can this connection be further formalized and used? 
	
	\item It is not clear to us to which extent are the decomposition theorems used for the barrier results (for homogeneous and multilinear polynomials) are tight. While their cost is relatively low, their tightness we feel is of independent interest. For examples, the decomposition of Lemma~\ref{lem:poly-rank} shows us that any 
	homogeneous polynomial matrix $M(\bx)$ of degree $d$ and rank $r$ has homogeneous 
	rank bounded by $(d+1) \cdot r$. On the other hand, Derksen and Makam prove 
	in~\cite{DM16} that there exists a linear matrix $M(\bx)$ of rank $r$ whose 
	homogeneous rank is lower bounded by $(2-\epsilon) \cdot r$. It is therefore 
	an interesting problem to decide if our decomposition is actually tight.
	More precisely, we make the following conjecture: 

	\begin{conjecture}
		For every $d$ and for every $\eps$ there exist a matrix $M(\bx)$ of homogeneous 
		polynomials	of degree $d$ and rank $r$ such that $\hrk(M(\bx)) > (d+1-\eps) r$. 
	\end{conjecture} 
\end{enumerate}

\section*{Acknowledgments}

The authors would like to thanks Mrinal Kumar and Zeev Dvir for helpful discussions 
during the preparation of this manuscript.

Part of Ankit's research was done when the author was a student at Princeton University. Research supported by Mark Braverman's NSF grant CCF-1149888, Simons Collaboration on Algorithms and Geometry, Simons Fellowship in Theoretical Computer Science and Siebel Scholarship.

Part of Rafael's research was done when the author was a student at 
Princeton University.
Research supported by NSF CAREER award DMS-1451191, NSF grant CCF-1523816 and Siebel Scholarship.

Part of Klim's research was done when the author was a post-doc in Tel-Aviv University. Research was supported   from the European Community's Seventh Framework Programme (FP7/2007-2013) under grant agreement number 257575.

Avi's research was supported by NSF grant CCF-1412958.

\bibliographystyle{alpha}
\bibliography{refs}

\end{document}